\documentclass[twocolumn]{IEEEtran}
\usepackage{amsmath,amssymb,amsthm,epsfig,color,subfigure,empheq,graphicx,graphics,pgfplots,balance}
\usepackage{enumerate,url,algorithm,algorithmic,wasysym,epstopdf,enumitem}
\usepackage{multirow} 
\usetikzlibrary{decorations.pathreplacing}
\usepackage{xcolor}
\usepackage{amsfonts}
\usepackage{xparse}
\newsavebox{\fminipagebox}
\NewDocumentEnvironment{fminipage}{m O{\fboxsep}}
 {\par\kern#2\noindent 
  \vspace{-3ex}
 \begin{center}\begin{lrbox}{\fminipagebox}
  \begin{minipage}{#1}\ignorespaces}
 {\end{minipage}\end{lrbox}%
  \makebox[#1]{%
    \kern\dimexpr-\fboxsep-\fboxrule\relax
    \fbox{\usebox{\fminipagebox}}%
    \kern\dimexpr-\fboxsep-\fboxrule\relax
  }\par\kern#2 \end{center}
  \vspace{-1ex}
 }

\newcounter{ProblemCounter}
\renewcommand\theProblemCounter{\arabic{ProblemCounter}}

\DeclareMathOperator{\range}{range}

\newtheorem{proposition}{Proposition}

\newtheorem{corollary}{Corollary}

\newcommand \bzero{\mathbf{0}}
\newcommand \bone{\mathbf{1}}
\newcommand \ba{\mathbf{a}}

\newcommand \bef{\mathbf{f}} 

\newcommand \bp{\mathbf{p}}
\newcommand \bq{\mathbf{q}}

\newcommand \bv{\mathbf{v}}

\newcommand \bA{\mathbf{A}}

\newcommand \bP{\mathbf{P}}
\newcommand \bQ{\mathbf{Q}}

\newcommand \btheta{\boldsymbol{\theta}}

\newcommand \bomega{\boldsymbol{\omega}}

\newcommand \mcE{\mathcal{E}}

\newcommand \mcG{\mathcal{G}}

\newcommand \mcN{\mathcal{N}}

\newcommand \mcS{\mathcal{S}}
\newcommand \mcT{\mathcal{T}}

\newcommand \bmcE{\bar{\mathcal{E}}}

\newcommand \tbA{\tilde{\mathbf{A}}}

\newcommand \cbA{\check{\mathbf{A}}}

\begin{document}

\title{Joint Grid Topology Reconfiguration \\
and Design of Watt-VAR Curves for DERs}

\author{	Manish K. Singh,~\IEEEmembership{Student Member,~IEEE},
			Sina Taheri,~\IEEEmembership{Student Member,~IEEE}, 
			Vassilis Kekatos,~\IEEEmembership{Senior Member,~IEEE},
			Kevin P. Schneider,~\IEEEmembership{Senior Member,~IEEE}, and
			Chen-Ching Liu,~\IEEEmembership{Fellow,~IEEE} 
		
}	
	

\maketitle
\begin{abstract}
Operators can now remotely control switches and update the control settings for voltage regulators and distributed energy resources (DERs), thus unleashing the network reconfiguration opportunities to improve efficiency. Aligned to this direction, this work puts forth a comprehensive toolbox of mixed-integer linear programming (MILP) models leveraging the control capabilities of smart grid assets. It develops detailed practical models to capture the operation of locally and remotely controlled regulators, and customize the watt-var DER control curves complying with the IEEE 1547.8 mandates. Maintaining radiality is a key requirement germane to various feeder optimization tasks. This requirement is accomplished here through an intuitive and provably correct formulation. To the best of our knowledge, this is the first time to optimally select a feeder topology and simultaneously design DER settings while taking into account legacy grid apparatus. The developed toolbox is put into action to reconfigure a grid for minimizing losses using real-world data on a benchmark feeder. The results corroborate that optimal topologies vary across the day and coordinating DERs and regulators is critical during periods of steep net load changes.
\end{abstract}


\section{Introduction}
\allowdisplaybreaks
Power distribution grids, in general, operate as radial networks connecting the substations to various customers. Oftentimes, these systems host normally-open switches that allow changes in the network topology and maintain radiality for protection system simplicity. The ability to switch between different topologies brings about a class of grid optimization tasks termed as \emph{distribution network reconfiguration}~(DNR). Some goals of DNR are post-outage restoration, load balancing, voltage regulation, power loss minimization, and planned maintenance~\cite{jianhui2017proc}, \cite{BW3}, \cite{JDLS19}, \cite{KA14}.

Utilizing existing switches to enhance efficiency and reliability of distribution systems is promising, making DNR a long pursued task~\cite{BW3}, \cite{LLV88}, \cite{LLV89}. Network reconfiguration problems are combinatorial in nature, and inevitably introduce integer variables when posed as mathematical programs. However, advancements in mixed-integer solvers for linear, quadratic, and second-order cone programs revived attempts towards efficient DNR reformulations~\cite{TH12}, \cite{JSP12}. Meanwhile, the development of exact conic relaxations of optimal power flow and accurate three-phase linear models have enabled computationally scalable DNR approaches that can cater to unbalanced multiphase grids~\cite{JDLS19}, \cite{LiuLi2019SDP}.

The advent of distributed energy resources (DERs) such as small generators, microgrids, and flexible loads has directed recent DNR research at maximally utilizing the available infrastructure~\cite{ccliu2018coordinating}, \cite{jianhui2018sequential}, \cite{SKL19GM}. On the other hand, the intermittency introduced by DERs increases the importance of DNR towards maintaining voltages within safe limits~\cite{Ghamsari2016}. Thus, attempts are being directed towards leveraging smart grid assets such as dispatchable DERs, capacitor banks, and remotely controlled voltage regulators in DNR formulations~\cite{LiuLi2019SDP},\cite{JDLS19}. 

Nonetheless, several smart grid devices (such as inverter-based photovoltaics (PVs) or energy storage units) and legacy grid devices alike operate based on local control rules~\cite{IEEE1547}, \cite{hiskens2016tap}. On an operational basis, these rules could be fixed (regulators and capacitor banks) or reconfigured periodically~\cite{Jalali19}. This is to reduce the frequency in communication and optimal power flow computations~\cite{KeWaCoGia15}. Yet the outcome of DNR could be significantly affected by inaccurate or inadequate modeling of such locally controlled devices. The attempts at proper modeling of these devices are limited and based on simplifying assumption,s such as fixed and known taps for regulators and unity power factor DERs~\cite{jianhui2018sequential}.

Enforcing radiality is another critical aspect in grid topology reconfiguration and other optimization tasks, such as planning and topology identification~\cite{COMH15}, \cite{TKC19}. Popular approaches to enforce radiality include an exhaustive loop elimination, imposing a single inflow edge or a single parent per bus; see \cite{LCSH19}, \cite{Ahmadi15}, and references therein. Despite being a classical problem, the conventional approaches for enforcing radiality fail or lack optimality guarantees in the presence of DERs~\cite{LCSH19}.

The contribution of this work is threefold: \emph{i)} Put forth a novel mixed-integer linear program (MILP) model for designing watt-var curves for DERs that takes into account all IEEE 1547.8 standard mandates (Section~\ref{sec:pv}); \emph{ii)} Revisit an optimization model for guaranteeing connectedness and radiality of a feeder to provide a more compact form and establish its correctness (Section~\ref{sec:radial}). The model is intuitive, provably correct, and decouples radiality constraints from variables capturing actual flows; and \emph{iii)} To capture the effect of legacy devices, we develop an optimization model for capturing the operation of locally controlled regulators (Section~\ref{sec:vr}). This is in contrast to existing schemes where regulators are either ignored or their taps are presumed known. The proposed DNR is formulated as a mixed-integer quadratic program (MIQP) and tested using real-world load and solar generation data on the IEEE 37-bus benchmark. The tests of Section~\ref{sec:tests} corroborate that depending on the load-generation mix experienced across a day, the operator has to select different topologies as well as regulator and DER settings. Albeit our results build on a linearized and balanced grid model, they constitute a solid foundation for extensions to AC models and unbalanced multiphase setups. 

Regarding \emph{notation}, lower- (upper-) case boldface letters denote column vectors (matrices). Calligraphic symbols are reserved for sets. Symbol $^{\top}$ stands for transposition, and vectors $\mathbf{0}$ and $\mathbf{1}$ are the all-zero and all-one vectors. 

\section{Problem Statement and Existing Models}\label{sec:PF}
Suppose a utility knows the model of a feeder as well as the anticipated load and solar generation on a per-bus basis for the upcoming operating period of 4 hours or so. The utility operator would like to reconfigure the grid via remotely controlled switches to minimize ohmic losses. A key requirement is that the reconfigured topology has to remain radial at all times. In addition to switches, the operator can change the tap settings of remotely controlled voltage regulators and select the watt-var curves of DERs to ensure that voltages and line flows remain within specified limits. While selecting the feeder topology and optimizing the settings of regulators and PVs, the operator has to further take into account non-controllable loads and legacy devices.

To tackle this problem, this section reviews optimization models for feeders, nodal and edge constraints, and voltage-dependent loads. It should be emphasized that these are existing models. They are presented here for completeness and for setting the ground of subsequent developments: Section~\ref{sec:pv} puts forth a novel approach for designing watt-var curves for PV generators. Section~\ref{sec:radial} devises an efficient optimization model for enforcing radiality. Section~\ref{sec:vr} presents models for locally and remotely controlled voltage regulators. Building on the previous models, Section~\ref{sec:pf} formulates DNR and the numerical tests of Section~\ref{sec:tests} validates the method.

\subsection{Grid Modeling, Nodal Variables and Constraints}\label{subsec:nodal}
Before commencing with the feeder models, some preliminaries from graph theory are in order. An undirected graph $\mcG:=(\mcN,\mcE)$ is defined by a set of nodes $\mcN$ and a set of edges $\mcE$, that are incident on the nodes in $\mcN$. Any edge $e\in\mcE$ is defined by its incident nodes as $(i,j)$ with $i,~j\in\mcN$. Nodes $i$ and $j$ are said to be \emph{adjacent} if there is an edge $(i,j)$ or $(j,i)$ in $\mcE$. Edges $e_1$ and $e_2$ are adjacent if they have a common end node. A \emph{path} from node $i$ to $j$ is a sequence of adjacent edges, without repetition, starting from $i$ and terminating at $j$, such that no node is revisited. A graph $\mcG$ is \emph{connected} if there exists an $i-j$ path for all $i,~j\in\mcN$. A \emph{cycle} is a sequence of adjacent edges without repetition that starts and ends at the same node. A graph with no cycles is \emph{acyclic}. A connected and acyclic graph is a \emph{tree}. A graph $\check{\mcG}:=(\check{\mcN},\check{\mcE})$ is a subgraph of $\mcG$ if $\check{\mcN}\subseteq\mcN$ and $\check{\mcE}\subseteq\mcE$. If every edge $e\in\mcE$ is assigned a direction, the obtained graph is termed \emph{directed}.

A single-phase distribution system with $N+1$ buses can be modeled as a connected graph $\mcG(\mcN_0,\mcE)$. The nodes in $\mcN_0:=\{0,\dots,N\}$ correspond to buses; and its edges $\mcE$ to distribution lines, voltage regulators, and switches. The substation bus is indexed by $i=0$, and other buses are contained in ${\mcN:=\mcN_0\setminus\{0\}}$. The assumption of a single feeder bus (substation) is without loss of generality. The detailed generalization will be commented upon at various instances while declaring constraints relating to feeder bus. Topologically, a graph representing instances of multiple substations may be augmented by appending a virtual bus that is connected to all substations, thus acting as a single substation bus. Each edge $e=(i,j)$ is assigned a direction from the origin node $i$ to the destination node $j$. If $(i,j)\in\mcE$, then $(j,i)\notin\mcE$. 

Each bus $i\in\mcN$ is assumed to host at most one generator or load. The subset of buses hosting loads is denoted by $\mcN_\ell\subseteq\mcN$. This is without loss of generality because a bus with multiple loads and/or generations can be modeled as a set of single-load buses, all connected by non-switchable zero-impedance lines. Let $v_i$ represent the voltage magnitude and $p_i+jq_i$ the complex power injection on bus $i$. The nodal voltages and injections at all nodes in $\mcN$ can be stacked in the $N$-length vectors $\bv$ and $\bp+j\bq$, respectively. The substation voltage $v_0$ is assumed known and fixed. The general case of multiple substations can be handled by defining voltages independently for all substation buses connected to bus $0$. We do not consider this scenario to keep the presentation uncluttered.

A distribution grid may host different types of loads and DERs. Some examples include (in)elastic ZIP loads; (non)dispatchable DERs; and volt- or watt-dependent reactive power sources per the IEEE 1547.8 standard~\cite{IEEE1547}. The constraints on voltage and power injection for all nodes can be abstractly expressed as
\begin{subequations}\label{eq:limits}
	\begin{align}
	\underline{v}\bone&\leq \bv\leq \bar{v}\bone\label{seq:vol_lim}\\
	\underline{\bp}(\bv)&\leq \bp\leq \bar{\bp}(\bv) \label{seq:p_lim}\\
	\underline{\bq}(\bv,\bp)&\leq \bq\leq \bar{\bq}(\bv,\bq).\label{seq:q_lim}
	\end{align}
\end{subequations}
The individual limits are discussed next. The voltage limits may be set to the typical operational limits: The ANSI standard dictates that service voltages should remain within $\pm 5\%$~per unit (pu) \cite{ansic84}. Our model stops at the level of distribution transformers. Expecting a voltage deviation along the cable between a distribution transformer and the service voltage, the practice is to maintain voltages at distribution transformers within $\pm 3\%$~pu; see also \cite{Kersting}, \cite{Jalali19}.

The functions $\underline{\bp}(\bv)$, $\bar{\bp}(\bv)$, $\underline{\bq}(\bv,\bp)$, and $\bar{\bq}(\bv,\bq)$ apply entrywise, and depend on load and DER characteristics. Regarding loads, in steady-state analysis the voltage dependence of loads is captured by the ZIP model. According to this model, each load is a composition of a constant-impedance (Z), a constant-current (I), and a constant-power (P and Q) component. Given bus voltage magnitude $v_i$, the power injection of load $i$ is modeled as~\cite{Kersting}
\begin{align*}
p_i(v_i)&=\alpha_0^p+\alpha_1^pv_i+\alpha_2^pv_i^2\\
q_i(v_i)&=\alpha_0^q+\alpha_1^qv_i+\alpha_2^qv_i^2
\end{align*}
with all $\alpha$ coefficients being non-positive and assumed known. 

Because under normal operation voltages are close to $1$~pu, we can linearize the quadratic dependence of ZIP loads around the nominal voltage to approximate $v_i^2 \simeq 2v_i-1$; see e.g., \cite{FCL}. Then, for all buses hosting loads, the active and reactive power limits of \eqref{seq:p_lim}--\eqref{seq:q_lim} can be compactly written as
\begin{equation}\label{eq:ZIP}
[\underline{p}_i(v_i)~~\bar{p_i}(v_i)~~\underline{q}_i(v_i)~~\bar{q}_i(v_i)]^\top = \boldsymbol{\alpha_0} +v_i \boldsymbol{\alpha_{12}},~\forall i\in\mcN_\ell.
\end{equation}
If load $i$ is inelastic, then apparently $\underline{p}_i(v_i)=\bar{p_i}(v_i)$; and $\underline{p}_i(v_i)\leq \bar{p_i}(v_i)$ otherwise. Similarly for reactive power injections. Modeling of DERs is deferred to Section~\ref{sec:pv}.

\subsection{Edge Variables and Constraints}\label{sec:edge}
The edge set $\mcE$ can be partitioned into the set of switches $\mcE_S$; regulators $\mcE_R$; and fixed lines $\mcE\setminus(\mcE_R\cup\mcE_S)$. The basic network reconfiguration task aims at selecting a subset of switches to be closed. To capture which switches are closed, let us introduce the binary variables $y_e$'s for all switchable lines $e\in\mcE_S$. Variable $y_e=1$ indicates that switch $e$ is closed or connected; and vice versa. 

Let the power flow on edge $e\in\mcE$ be $P_e+jQ_e$. The power flow constraints on distribution lines may be expressed as
\begin{subequations}\label{eq:flim}
	\begin{align}
	[\underline{P}_e~\underline{Q}_e]&\leq [P_e~Q_e]\leq[\bar{P}_e~\bar{Q}_e],~\forall~e\in\mcE\setminus\mcE_S\label{seq:flim1}\\
	y_e[\underline{P}_e~\underline{Q}_e]&\leq [P_e~Q_e]\leq y_e[\bar{P}_e~\bar{Q}_e],~\forall~e\in\mcE_S.\label{seq:flim2}
	\end{align}
\end{subequations}
If switch $e$ is open ($y_e=0$), constraint \eqref{seq:flim2} sets the power flow on $e$ to zero. Else, box constraints on the power flow are enforced and usually $\underline{P}_e=-\bar{P}_e$ and $\underline{Q}_e=-\bar{Q}_e$. Although apparent power flow limits of the form $P_e^2+Q_e^2\leq S_e^2$ can be added to our formulation, they result in a mixed-integer quadratically-constrained quadratic program (MI-QCQP), which does not scale as well as an MIQP. Alternatively, apparent power constraints on flows can be handled by a polytopic approximation of $P_e^2+Q_e^2\leq S_e^2$; see \cite{Jabr18}. This approach is not taken here for clarity of presentation.

\subsection{Power Flow Model}\label{subsec:power}
To relate power injections and flows to voltages, we build upon the \emph{linearized distribution flow} (LDF) model of \cite{BW3}. Albeit approximate, the LDF model has been extensively employed for various grid optimization tasks with satisfactory accuracy~\cite{BoDo15}. By ignoring power losses, LDF postulates that the power injections for each bus $i\in \mcN$ are
\begin{subequations}\label{eq:ldf1}
	\begin{align}
	&p_i = \sum_{e:(i,j) \in \mcE}P_e - \sum_{e:(j,i) \in \mcE}P_e \label{seq:active}\\
	&q_i = \sum_{e:(i,j) \in \mcE}Q_e - \sum_{e:(j,i) \in \mcE}Q_e.\label{seq:reactive}
	\end{align}
\end{subequations}

If $r_e+jx_e$ represents the impedance of line $e:(i,j)\in\mcE$, the LDF model relates the \emph{squared} voltage magnitudes to power flows linearly as $v_i^2-v_j^2=2r_eP_e+2x_eQ_e$. Invoking the assumption of small voltage deviations, squared voltages can be approximated as $v_i^2\simeq 2v_i-1$. Then, the non-squared voltages can be substituted in the LDF model to yield 
\begin{equation}\label{eq:ldf2}
v_i-v_j=r_eP_e+x_eQ_e
\end{equation}
for each line $e:(i,j)\in\mcE\setminus(\mcE_R\cup\mcE_S)$. The approximate voltage drop model of \eqref{eq:ldf2} can be alternatively derived by linearizing the power flow equations at the flat voltage profile~\cite{BoDo15}, \cite{Deka1}, \cite{TJKT20}. 

For switchable lines in $\mcE_S$, the voltage drop of \eqref{eq:ldf2} applies only if the switch is closed, that is
\begin{equation}\label{eq:ldf3}
y_e(v_i-v_j-r_eP_e-x_eQ_e)=0,\quad\forall~e:(i,j)\in\mcE_S.
\end{equation} 
The bilinear products appearing in \eqref{eq:ldf3} such as $y_ev_i$ are handled using \emph{McCormick linearization}, which is briefly reviewed next.

McCormick linearization replaces the product of variables by their linear convex envelopes to yield a relaxation of the original non-convex feasible set~\cite{McCormick1976}. 
If at most one of the factor variables is continuous and the rest are binary, the relaxation becomes \emph{exact}. Take for instance the product $z=xy$ over a binary variable $x\in\{0,1\}$, and a continuous variable $y$ bounded within $y\in[\underline{y},\bar{y}]$. The constraint $z=xy$ can be equivalently expressed as four linear equality constraints
\begin{subequations}\label{eq:MC}
	\begin{align}
	x\underline{y}&\leq z\leq x\bar{y},\label{seq:MC1}\\
	y+(x-1)\bar{y}&\leq z\leq y +(x-1)\underline{y}\label{seq:MC2}.
	\end{align}
\end{subequations}
To see the equivalence, evaluate $x=0$ in \eqref{eq:MC} to get $z=0$, and evaluate $x=1$ to get $z=y$. For a continuous-binary bilinear product, the McCormick linearization is equivalent to the so called big-$M$ trick. However, particular emphasis on tight bounds $y\in[\underline{y},\bar{y}]$ in McCormick linearization tends to provide numerical superiority. All continuous-binary bilinear products encountered henceforth will be handled by McCormick linearization. The resulting linear inequalities of \eqref{eq:MC} will not be provided explicitly for brevity. 


\section{Designing Watt-Var Control Curves for DERs}\label{sec:pv}
This section specifies the power injection limits of \eqref{seq:p_lim}--\eqref{seq:q_lim} for DERs. Conventionally, DERs have been modeled as constant-power sources operating at unit power factor~\cite{Turitsyn11}. With smart DERs featuring enhanced sensing, communication, and actuation, the IEEE 1547.8 standard mandates DERs to provide reactive power support~\cite{IEEE1547}. According to the standard, the reactive power injection of DERs can follow four possible modes~\cite{IEEE1547}: \emph{i)} constant power factor; \emph{ii)} voltage-dependent reactive power (volt-var); \emph{iii)} active power-dependent reactive power (watt-var); and \emph{iv)} constant reactive power mode. 

The volt-var and watt-var dependencies are captured by control rules described by piecewise affine functions; see Fig.~\ref{fig:curve}. The operator may change these rules on a daily, hourly, or near-real-time basis. To effectively integrate DERs, their control rules should be decided optimally based on feeder and loading conditions. To this end, it is henceforth assumed that DERs are operating in the watt-var mode and their parameters are selected and kept fixed per periods of say 4 hours. 

\begin{figure}[t]
	\centering
	\includegraphics[scale=0.35]{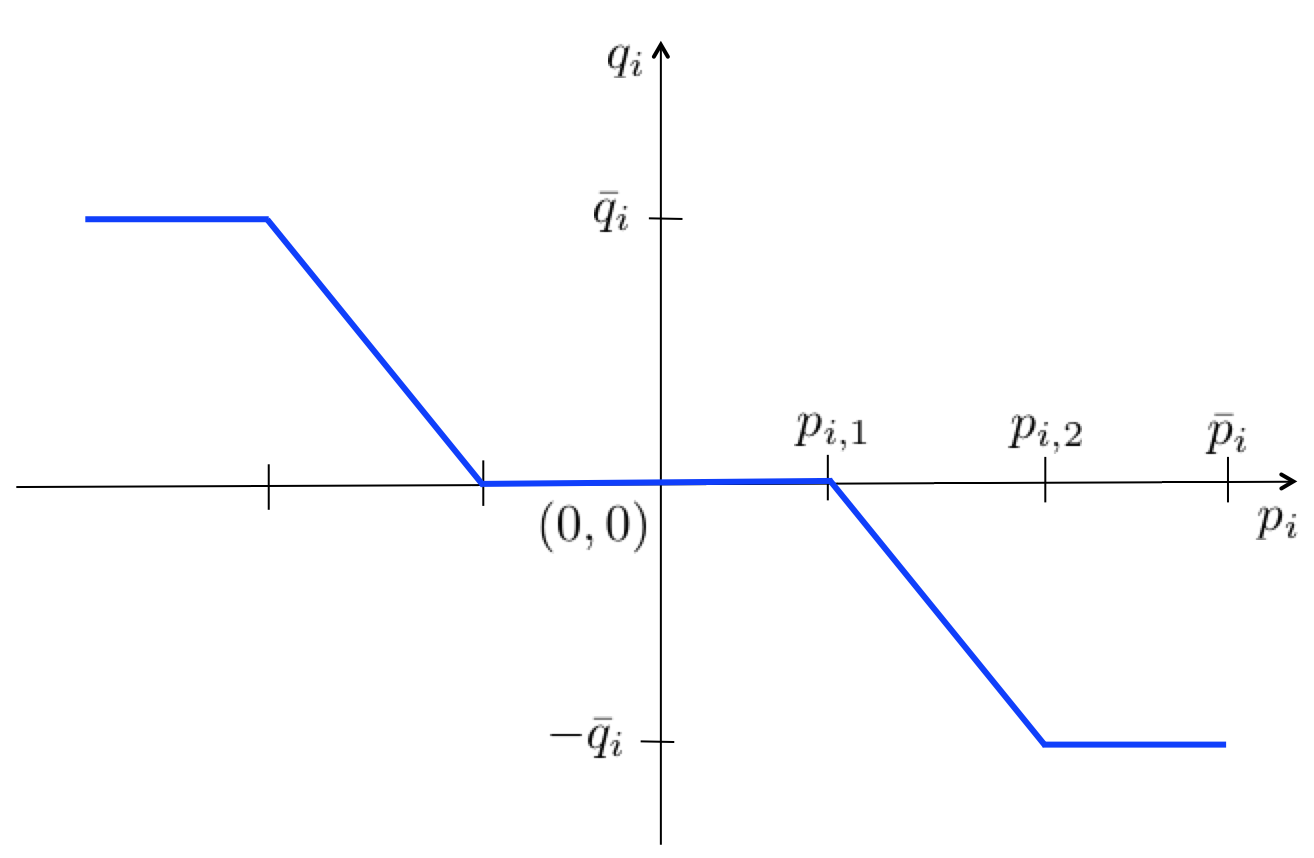}
	\caption{Active power-reactive power (watt/var) DER control curve~\cite{IEEE1547}.}
	\label{fig:curve}
\end{figure}

The watt-var inverter control is implemented via the piecewise affine rules of Fig.~\ref{fig:curve}. The left half applies to DERs featuring active power absorption, such as energy storage units. To simplify the exposition, we consider DERs operating in the right halfspace of the watt-var rule, that is DER that only inject active power to the feeder (e.g., renewable generation). Given the rated reactive power capacity $\bar{q}_i$ for the $i$-th DER, the controllable parameters are $p_{i,1}$ and $p_{i,2}$. The IEEE 1547.8 standard further constraints $(p_{i,1},p_{i,2})$ so that
\begin{subequations}\label{eq:plim}
\begin{align}	
0.4\bar{p}_i & \leq p_{i,1}\leq 0.8 \bar{p}_i\label{seq:plim1}\\
p_{i,1} + 0.1\bar{p}_i & \leq p_{i,2}\leq \bar{p}_i\label{seq:plim2}
\end{align}
\end{subequations}
where $\bar{p}_i$ is the rated active power for DER $i$. These specifications are set by the standard to ensure a substantial deadband and to avoid steep slopes in Fig.~\ref{fig:curve}.

Given $(p_{i,1},p_{i,2})$, the reactive power injection of DER $i$ depends on its active power injection as
\begin{equation}\label{eq:rule}
q_i(p_i)=\left\{\begin{array}{ll}
0&,~ 0\leq p_i\leq p_{i,1}\\
\frac{-\bar{q}_i(p_i-p_{i,1})}{p_{i,2}-p_{i,1}} &,~ p_{i,1}\leq p_i\leq p_{i,2}\\
-\bar{q}_i &,~ p_{i,2}\leq p_i\leq \bar{p}_i
\end{array}\right.
\end{equation}
The control rule of \eqref{eq:rule} induces a non-linear equality constraint between the optimization variables $q_i$, $p_{i,1}$, and $p_{i,2}$. We next capture this constraint via a novel MILP model. 

Let us introduce the binary variables $(\delta_{i,1},\delta_{i,2},\delta_{i,3})$ to pick which of the three segments in \eqref{eq:rule} is active each time:
\begin{subequations}\label{eq:delta}
\begin{align}	
&(\delta_{i,1},\delta_{i,2},\delta_{i,3})\in\{0,1\}^3 \label{eq:delta1}\\
&\delta_{i,1}+\delta_{i,2}+\delta_{i,3}=1.\label{eq:delta2}
\end{align}
\end{subequations}
The selection of a segment depends on the value of $p_i$ as
\begin{equation}\label{eq:pdelta}
\delta_{i,2}p_{i,1} + \delta_{i,3}p_{i,2} \leq p_i \leq \delta_{i,1}p_{i,1} + \delta_{i,2}p_{i,2} + \delta_{i,3}\bar{p}_i.
\end{equation}
Then, the rule of \eqref{eq:rule} can be expressed by the constraint
\begin{equation}\label{eq:rule2}
q_i=\delta_{i,1}\cdot 0 - \delta_{i,2} \frac{\bar{q}_i(p_i-p_{i,1})}{p_{i,2}-p_{i,1}} -\delta_{i,3}\bar{q}_i.
\end{equation}
Although \eqref{eq:pdelta} involves binary-continuous variable products, and can be thus handled by MacCormick linearization, that is not the case for \eqref{eq:rule2}. Unfortunately, the latter entails ratios or products among continuous variables. 

To bypass this difficulty, the key idea here is to parameterize Figure~\ref{fig:curve} using the slope/intercept of its middle segment instead of the breakpoints $(p_{i,1},p_{i,2})$. If the middle segment of \eqref{eq:rule} is denoted by $q_i(p_i)=\beta_ip_i+\gamma_i$ for some negative $(\beta_i,\gamma_i)$, then \eqref{eq:rule2} is equivalent to
\begin{equation}\label{eq:rule3}
q_i=\delta_{i,2}(\beta_ip_i+\gamma_i)-\delta_{i,3}\bar{q}_i.
\end{equation}
Different from \eqref{eq:pdelta}, constraint \eqref{eq:rule3} involves only binary-continuous variable products. 

We next reformulate \eqref{eq:pdelta} in terms of $(\beta_i,\gamma_i)$. Because the line $q_i(p_i)=\beta_ip_i+\gamma_i$ passes through the points $(p_{i,1},0)$ and $(p_{i,2},-\bar{q}_i)$, we get that
\begin{equation}\label{eq:pbg}
p_{i,1}=-\frac{\gamma_i}{\beta_i} \quad \quad \text{and}\quad \quad p_{i,2}=-\frac{\bar{q}_i+\gamma_i}{\beta_i}.
\end{equation}
Plugging \eqref{eq:pbg} into \eqref{eq:pdelta}; multiplying all sides by $\beta_i<0$; adding $\gamma_i$; and using \eqref{eq:delta2}, eventually provides 
\begin{equation}\label{eq:pbg2}
\delta_{i,3}(\gamma_i-\beta_i\bar{p}_i)-\delta_{i,2}\bar{q}_i\leq\beta_ip_i+\gamma_i \leq\delta_{i,1}\gamma_i -\delta_{i,3}\bar{q}_i
\end{equation}
which is still amenable to McCormick linearization.

The control rule of \eqref{eq:rule} is equivalent to \eqref{eq:delta}, \eqref{eq:rule3}, and \eqref{eq:pbg2}. With the help of McCormick linearization, the latter can be posed as an MILP model. The aforesaid model captures the piecewise control rule, but do not enforce the limitations of \eqref{eq:plim}. To capture these IEEE 1547.8 requirements, we will translate the constraints on $(p_{i,1},p_{i,2})$ to constraints on $(\beta_i,\gamma_i)$. Substituting \eqref{eq:pbg} into \eqref{eq:plim} implies that $(\beta_i,\gamma_i)$ should satisfy
\begin{subequations}\label{eq:wvar2}
\begin{align}	
-0.4\bar{p}_i\beta_i&\leq \gamma_i \leq -0.8\bar{p}_i\beta_i\\
\bar{p}_i\beta_i+\gamma_i&\leq -\bar{q}_i\leq 0.1\bar{p}_i\beta_i.
\end{align}
\end{subequations}
To summarize, the control curve for DER $i$ is optimally tuned via variables $(\beta_i,\gamma_i)$ that satisfy \eqref{eq:delta}, \eqref{eq:rule3}, \eqref{eq:pbg2}, and \eqref{eq:wvar2}. To the best of our knowledge, this is the first model to optimally design the IEEE 1547 control curves for DERs.

\section{Ensuring Radial Topologies}\label{sec:radial}
Ensuring a graph is radial is of central importance to various grid optimization tasks. In grids with a single power source and no DERs, enforcing radiality entails limiting the number of edges with incoming flow to one per bus~\cite{TH12}. In the presence of DERs, a bus may receive power from multiple edges even if the grid is radial. To handle such networks, the model of \cite{JSP12} enforces an edge orientation so that each bus has a single parent bus. Despite its extensive use in the grid topology reconfiguration/restoration literature, counterexamples where this parent-child model produces disconnected graphs do exist~\cite{Ahmadi15}. A dual graph-based model was suggested in \cite{Ahmadi15}, yet it is limited to planar graphs. For a general network, cycles can be avoided by imposing that the number of connected edges on each cycle to be less than the length of the cycle~\cite{SKL19GM}. Despite its generality, this cycle-elimination approach can lead to exponentially many constraints. One of the most popular radiality model ensures connectivity of loads to DERs via the power flow equations, and connects DERs to the substation via flows of a virtual commodity~\cite{LFRR12}. The tightness of a linear programming relaxation for this model has also been recently reported in~\cite{LCSH19}. In this section, we advance upon the commodity flow approach and propose a more succinct model with fewer variables and constraints. Moreover, this is the first time the commodity flow model is supported with a formal proof. 

Given the complete graph $\mcG(\mcN_0,\mcE)$, define a subgraph $\check{\mcG}(\mcN_0,\check{\mcE})$, such that $\check{\mcE}:=\mcE\setminus\{e:e\in\mcE_S, y_e=0\}$. The subgraph $\check{\mcG}$ represents the reconfigured distribution network. To capture the line infrastructure of $\mcG$, define its $|\mcE|\times(N+1)$ branch-bus incidence matrix $\tbA$ with entries
\begin{equation*}
\tilde{A}_{e,k}:=
\begin{cases}
+1&,~k=i\\
-1&,~k=j\\
0&,~\text{otherwise}
\end{cases}~\forall~e=(i,j)\in\mcE.
\end{equation*}
Separate the first column $\ba_0$ of $\tbA$ related to the substation bus $0$ as $\tbA=[\ba_0~\bA]$, to get the\emph{reduced} branch-bus incidence matrix $\bA$. Similarly, let $\cbA\in\mathbb{R}^{|\mcE'|\times N}$ represent the reduced branch-bus incidence matrix of subgraph $\mcG'$. The next claim establishes an efficient model for imposing graph connectivity.

\begin{proposition}\label{th:1}
A graph $\check{\mcG}(\mcN_0,\check{\mcE})$ with reduced branch-bus incidence matrix $\cbA$ is connected if and only if there exists a vector $\bef\in\mathbb{R}^{|\check{\mcE}|}$, such that 
\begin{equation}\label{eq:th1}
\cbA^\top\bef=\bone.
\end{equation}
\end{proposition}

\begin{proof}
Proving by contradiction, suppose $\check{\mcG}(\mcN_0,\check{\mcE})$ is not connected, and there exists $\bef\in\mathbb{R}^{|\check{\mcE}|}$ satisfying \eqref{eq:co}. If $\check{\mcG}(\mcN_0,\check{\mcE})$ is not connected, then there must exist a connected component, that is a maximal connected subgraph $\check{\mcG}_\mcS(\mcN_\mcS, \check{\mcE}_\mcS)$, such that $\mcN_\mcS\subset\mcN_0$ and $0\notin\mcN_\mcS$. Let $\bA_\mcS$ be the branch-bus incidence matrix of $\check{\mcG}_\mcS$. By definition, it holds that $\bA_\mcS\bone=\bzero$. The fundamental theorem of linear algebra implies
\begin{equation}\label{eq:range}
\bone\in(\range(\bA_\mcS^\top))^\perp\quad \text{or} \quad \bone\notin\range(\bA_\mcS^\top).
\end{equation}

By hypothesis, graph $(\mcN_\mcS, \check{\mcE}_\mcS)$ is a maximal connected subgraph of $\check{\mcG}$, and hence, there is no edge $(i,j)\in\check{\mcE}$ with $i\in\mcN_\mcS$ and $j\in\bar{\mcN}_{\mcS}$ where $\bar{\mcN}_{\mcS}:=\mcN_0\setminus\mcN_\mcS$. Since the order of edges and nodes forming the rows and columns of $\cbA$ are arbitrary, without loss of generality, partition $\cbA$ as
\begin{equation*}
\cbA=
\begin{bmatrix}
\bA_{\bar{\mcS}}& \bzero\\
\bzero & \bA_\mcS
\end{bmatrix}.
\end{equation*}
Heed that $\bA_{\bar{\mcS}}$ is a \emph{reduced} branch-bus incidence matrix, whereas $\bA_\mcS$ is a complete branch-bus incidence matrix since $0\notin\mcN_\mcS$. Partitioning $\bef$ conformably to $\cbA$, equation \eqref{eq:th1} reads
$$\begin{bmatrix}
\bA_{\bar{\mcS}}^\top& \bzero\\
\bzero & \bA_\mcS^\top
\end{bmatrix}\begin{bmatrix}
\bef_{\bar{\mcS}}\\\bef_\mcS
\end{bmatrix}=\bone.$$
The second block implies that $\bA_\mcS^\top\bef_\mcS=\bone$, which contradicts \eqref{eq:range} and completes the proof.
\end{proof}

To provide some circuit theoretic intuition on Proposition~\ref{th:1}, vector $\bef$ represents flows on $\check{\mcE}$ resulting from a unit injection at all network nodes except for node $i=0$. For this flow setup to be feasible, there must be a withdrawal of $N$ units at node $0$, and the injection from all nodes in $\mcN$ must have a path to reach node $0$. Having all nodes connected to node $0$ entails a single connected component. It is worth emphasizing that variable $\bef$ does not relate to the actual line flows and is introduced only to enforce connectivity. 

The condition for connectedness of Proposition~\ref{th:1} is defined on matrix $\cbA$, which depends on the switch status variables $y_e$'s. Notice that $\cbA$ is derived from $\bA$ by removing the rows related to open switches. Therefore, condition \eqref{eq:th1} can be expressed with respect to the original matrix $\bA$, by forcing the virtual flows in $\bef$ to be zero for open lines. The following corollary establishes a convenient constraint for connectedness to be used later in our network reconfiguration problem. 

\begin{corollary}\label{co:1}
Let $\bA$ be the reduced branch-bus incidence matrix of $\mcG$, and $\check{\mcG}\subseteq\mcG$ be a subgraph defined by opening switches $\{e\in\mcE_S:y_e=0\}$. Subgraph $\check{\mcG}$ is connected if and only if there exists $\bef\in\mathbb{R}^{|\mcE|}$ such that
\begin{subequations}\label{eq:co}
	\begin{align}
	\bA^\top\bef&=\bone\label{seq:co1},~\text{and}\\
	-y_eN&\leq f_e \leq y_eN,\quad\forall e\in\mcE_S\label{seq:co2}.
	\end{align}
\end{subequations}
\end{corollary}

Constraint \eqref{seq:co2} implies that the virtual flows on open switches are zero, and bounds the flows on closed switches within $[-N,~N]$. Once a $\check{\mcG}$ is ensured to be connected, the requirement of radiality can be readily enforced as 
\begin{equation}\label{eq:tree}
\sum_{e\in\mcE_S}y_e=N-|\mcE\setminus\mcE_S|
\end{equation}
to ensure the total number of connected edges is $N$.

\section{Modeling Voltage Regulators}\label{sec:vr}
In addition to the DER control settings and its topology, the voltage profile of a feeder is dependent on voltage regulators. This section develops novel models for regulators, which are later used in our grid reconfiguration formulation. 

We model voltage regulators as ideal transformers. This is without loss of generality because the impedance of a non-ideal regulator can be modeled as a line connected in series with the ideal regulator. An ideal regulator scales its secondary-side voltage by $\pm10\%$ on increments of $0.625\%$ using tap positions~\cite{hiskens2016tap}. Consider a regulator modeled by edge $e:(i,j)\in\mcE_R$. Its voltage transformation ratio can be set to $1+0.00625\cdot t_e$, where $t_e\in\{0,\pm1,\ldots,\pm 16\}$ is its tap position. We consider two classes of regulators~\cite{Kersting}: 

\subsubsection{Locally controlled regulators} Collect such regulators in set $\mcE_R^L\subseteq\mcE_R$. A locally controlled regulator $e:(i,j)\in\mcE_R$ is programmed to maintain $v_j$ within a given range $[\underline{v}_j,\bar{v}_j]$. The regulator changes its taps after a time delay until $v_j$ is brought within $[\underline{v}_j,\bar{v}_j]$, unless an extreme tap position has been reached. Ignoring the time delay, this operation is illustrated in Figure~\ref{fig:dvr} and described by
\begin{equation}\label{eq:DVR}
v_j(v_i)=\left\{\begin{array}{ll}
1.1\cdot v_i&,~v_i\leq \frac{\underline{v}_j}{1.1}\\
\left[\underline{v}_j,\overline{v}_j\right]&,~ \frac{\underline{v}_j}{1.1}< v_i <\frac{\overline{v}_j}{0.9}\\
0.9\cdot v_i&,~v_i\geq \frac{\overline{v}_j}{0.9}
\end{array}\right..
\end{equation}
The first branch relates to the case where the primary voltage $v_i$ is quite low and even with $t_e=+16$, the secondary voltage $v_j=1.1\cdot v_i$ remains below $\underline{v}_j$. Likewise, the third branch relates to the case where the tap has reached its minimum of $t_e=-16$. Normal operation is captured by the second branch, where $v_j$ is successfully regulated within $[\underline{v}_j,\overline{v}_j]$. 

\begin{figure}[t]
	\centering
	\includegraphics[scale=0.43]{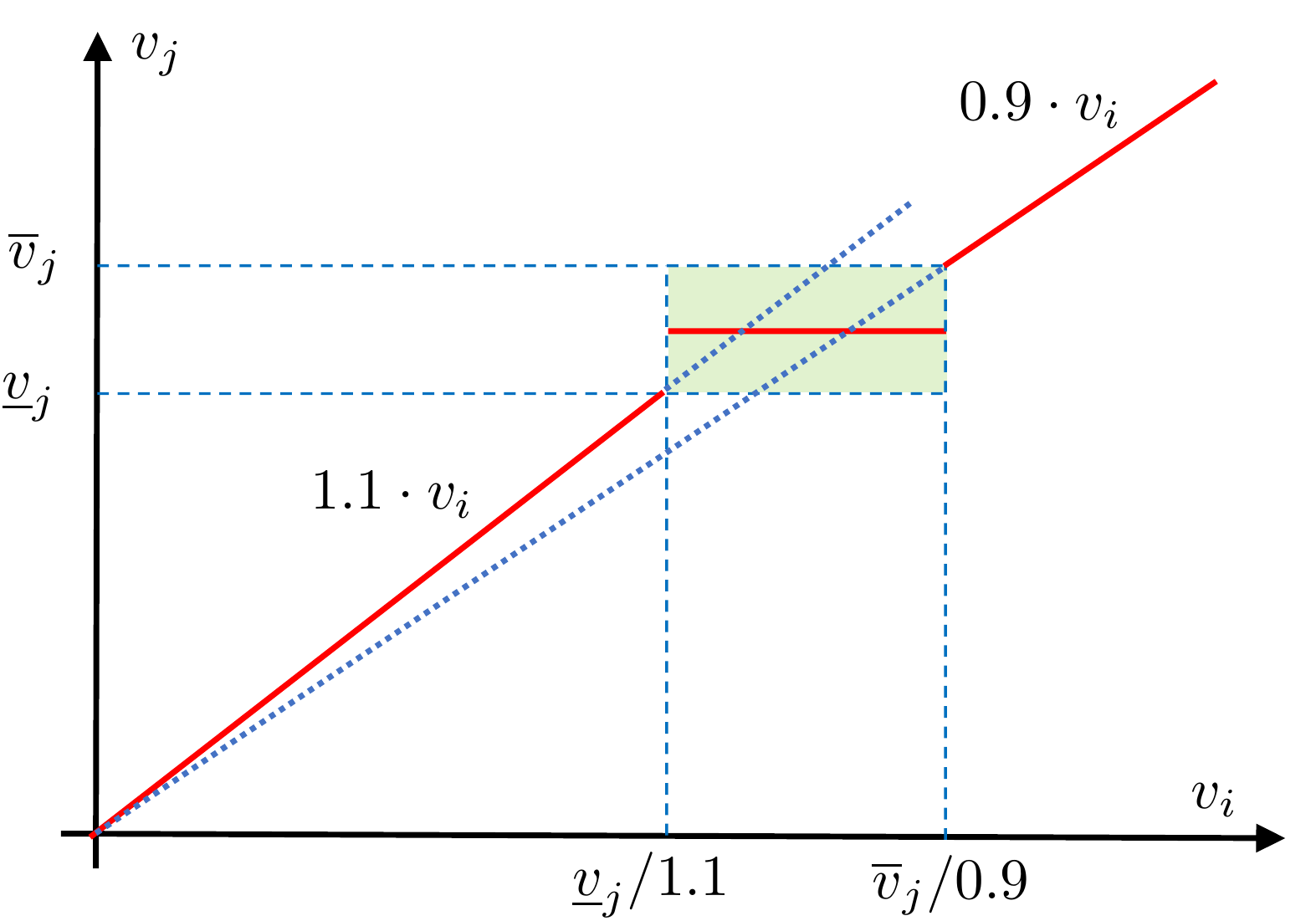}
	\caption{Locally controlled regulator characteristic: The left/rightmost segments occur when regulator taps have maxed out. Within the middle green box, the secondary voltage is successfully regulated. Lacking the actual tap position, this middle area is approximated by its midpoint (reference voltage).}
	\label{fig:dvr}
\end{figure}

For a locally controlled regulator, the operator cannot fully monitor and/or control the exact tap position. Because of this, we propose approximating the second branch of \eqref{eq:DVR} by setting $v_j$ at the mid-point of the regulation range, that is
\begin{equation*}
v_j(v_i)= \frac{\overline{v}_j+ \underline{v}_j }{2}\quad \text{when}\quad v_i\in\left(\frac{\underline{v}_j}{1.1},\frac{\overline{v}_j}{0.9}\right).
\end{equation*}
Since the regulation range typically spans $2$--$4$ taps or $0.0125$--$0.025$~pu~\cite{hiskens2016tap}, this approximation incurs negligible modeling error. The actual and approximate models for locally controlled regulators are illustrated in Fig.~\ref{fig:dvr}.

The regulator operation of Fig.~\ref{fig:dvr} can be modeled in a fashion similar to the watt-var curve of \eqref{eq:rule}. Let $\{\delta_{e,1},\delta_{e,2}, \delta_{e,3}\}$ be the binary variables selecting the three regions of operation. The approximate model for a locally controlled regulator $e\in\mcE_R^L$ is captured by the constraints
\begin{subequations}\label{eq:lvr}
	\begin{align}
	&v_i\geq0.8\delta_{e,1}+\frac{\underline{v}_j}{1.1}\delta_{e,2}+ \frac{\bar{v}_j}{0.9}\delta_{e,3} \label{seq:lvr1}\\
	&v_i\leq\frac{\underline{v}_j}{1.1}\delta_{e,1}+\frac{\bar{v}_j}{0.9} \delta_{e,2}+1.2\delta_{e,3} \label{seq:lvr2}\\
	&v_j=1.1\delta_{e,1}v_i+\delta_{e,2}\left(\frac{\overline{v}_j+\underline{v}_j}{2}\right)+0.9\delta_{e,3}v_i \label{seq:lvr4}\\
	&\delta_{e,1}+\delta_{e,2}+\delta_{e,3}=1 \label{seq:lvr3}\\
	&(\delta_{i,1},\delta_{i,2},\delta_{i,3})\in\{0,1\}^3.\label{seq:lvr5}
	\end{align}
\end{subequations}
Constraints \eqref{seq:lvr3}--\eqref{seq:lvr5} ensure that exactly one indicator variable gets activated. Constraints \eqref{seq:lvr1}--\eqref{seq:lvr2} capture the value of the primary voltage per region of Fig.~\ref{fig:dvr}. Constraint \eqref{seq:lvr4} captures the behavior of the secondary voltage per region of Fig.~\ref{fig:dvr}. The binary-continuous variable products in \eqref{seq:lvr4} can be handled via McCormick linearization. Note that $0.8$~pu and $1.2$~pu are arbitrarily chosen as the extreme voltage limits for defining the range of the primary voltage.

\subsubsection{Remotely controlled regulators} These regulators comprise the set $\bmcE_R^L:=\mcE_R\setminus\mcE_R^L$. If $e:(i,j)\in\bmcE_R^L$, its tap $t_e$ can be changed remotely by the operator. It hence becomes a control variable taking one of $33$ possible values. These values can be encoded using $6$~bits~\cite{WuTianZhang17}. For example, the binary code $100001$ corresponds to tap $t_e=+16$; code $010000$ corresponds to the neutral position $t_e=0$; and $000000$ to $t_e=-16$. Then, voltage $v_j$ relates to $v_i$ as
\begin{subequations}\label{eq:rvr}
	\begin{align}
	&v_j=\left(0.9+0.00625\cdot\sum_{k=0}^5b_{e,k} 2^k\right)v_i \label{seq:rvr1}\\
	&b_{e,k}\in\{0,1\},\quad k=0,\ldots,5\label{seq:rvr2}
	\end{align}
\end{subequations}
with the parenthesis being the binary encoding of tap $t_i$. The products $b_{e,k}v_i$ can be handled by McCormick linearization.

\section{Problem Formulation}\label{sec:pf}
Having modeled the major grid assets, we can now formulate the optimal grid reconfiguration task. 
Consider an operating period of $4$~hr. Before the start of this period, the operator collects minute-based data capturing the anticipated load and solar generation, and partition them into 15-min intervals. Then, from each 15-min interval, the operator selects $S$ samples, yielding a total of $T=16S$ samples for the upcoming 4-hr period. These samples are indexed by $t=1,\ldots,T$. Instead of sampling, the operator may use the averages per 15-min interval. The data related to sample $t$ are collectively denoted by vector $\btheta^t$. The operator would like to minimize the power losses summed up over all $T$ instances. Each one of the $T$ instances will be experiencing different power flow conditions. Nonetheless, all intervals share the same feeder topology, DER curves, and regulator settings. To capture this, we group optimization variables as
\begin{subequations}\label{eq:omega}
	\begin{align*}
	&\bomega_1:=\{\{y_e\}_{e\in\mcE_S},\{\beta_i,\gamma_i\}_{i\in\mcN\setminus\mcN_\ell},\{b_{e,k}\}_{e\in\mcE_R\setminus\mcE_R^L}\};~\text{and}\\
	&\bomega_2^t:=\{\bv^t,\bp^t,\{q_i^t,\delta_{i,k}^t\}_{i\in\mcN\setminus\mcN_\ell},\{\delta_{e,k}^t\}_{e\in\mcE_R^L},\bP^t,\bQ^t\},~\forall t.
	\end{align*}
\end{subequations}
The ultimate goal is to determine $\bomega_1$, that is a tree topology, inverter watt-var parameters, and regulator tap settings. The grid would then be allowed to operate autonomously using local rules per interval $t$ yielding variables $\{\bomega_2^t\}_{t=1}^T$.
 
The grid reconfiguration task can now be posed as
\begin{align*}\label{eq:DNR}
\min~&~\sum_{t\in\mcT}\sum_{e\in\mcE\setminus\mcE_R} r_e(P_{e,t}^2+Q_{e,t}^2)\tag{DNR}\\
\mathrm{over}~& ~ \bomega_1,\{\bomega_2^t\}_{t=1}^T\nonumber\\
\mathrm{s.to}~& ~\eqref{eq:limits}-\eqref{eq:ldf3},\eqref{eq:delta}, \eqref{eq:rule3}, \eqref{eq:pbg2}, \eqref{eq:wvar2},\eqref{eq:co},\eqref{eq:tree},\eqref{eq:lvr},\eqref{eq:rvr}~~\forall t.
\end{align*}
The cost function approximates the ohmic losses along all lines and times~\cite{BW3}, \cite{TH12}. When computing losses, only closed lines should be considered. However \eqref{eq:flim} entails that for open switches, the power flows are zero. This enables us to write the cost in \eqref{eq:DNR} regardless of the indicator variables $y_e^t$ for switchable lines.

\section{Numerical Tests}\label{sec:tests}
\begin{figure}[t]
	\centering
	\includegraphics[scale=0.3]{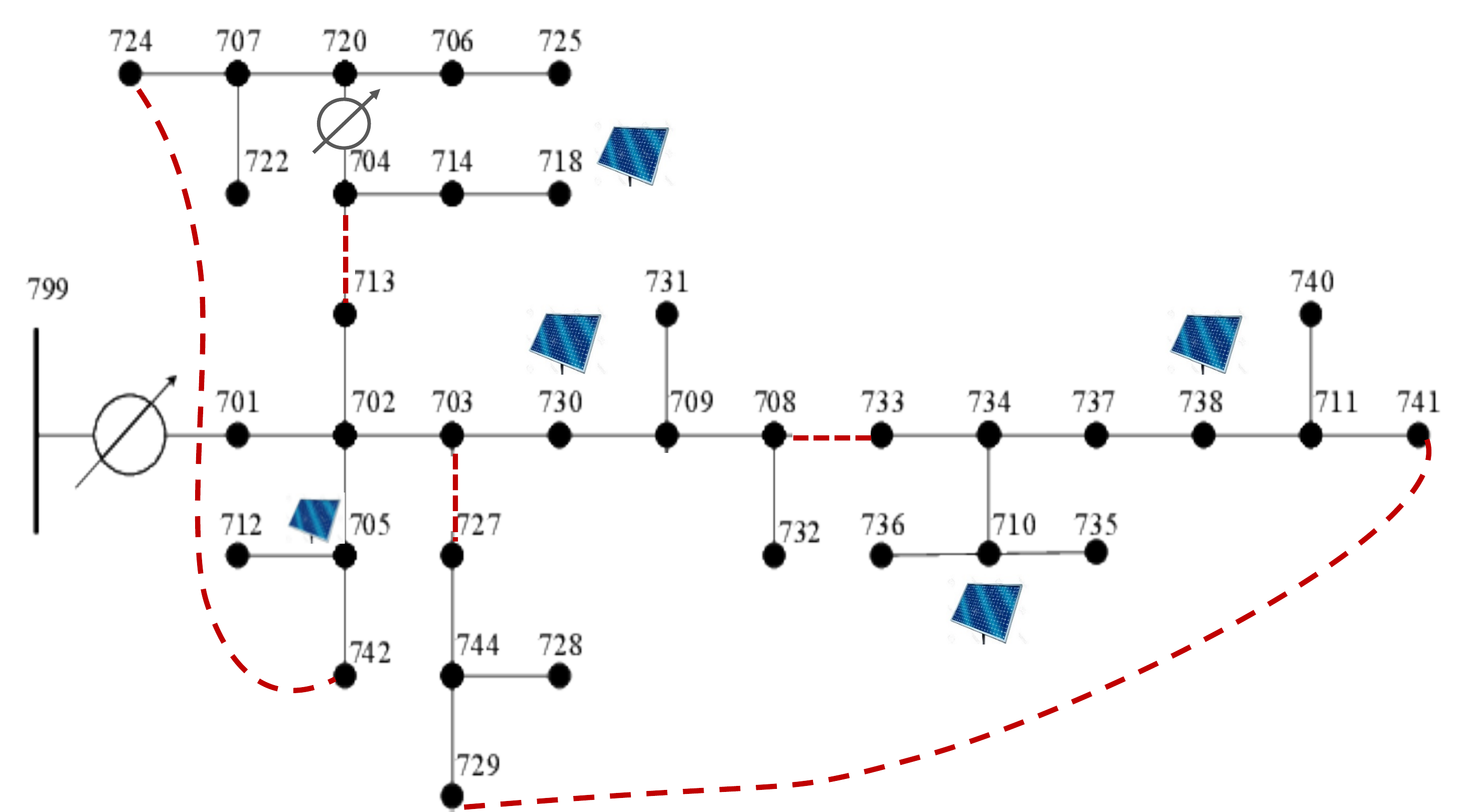}
	\caption{The IEEE 37-bus feeder with an additional regulator, lines, and DERs.}
	\label{fig:37bus}
\end{figure}

\begin{figure}[t]
	\centering
	\includegraphics[scale=0.29]{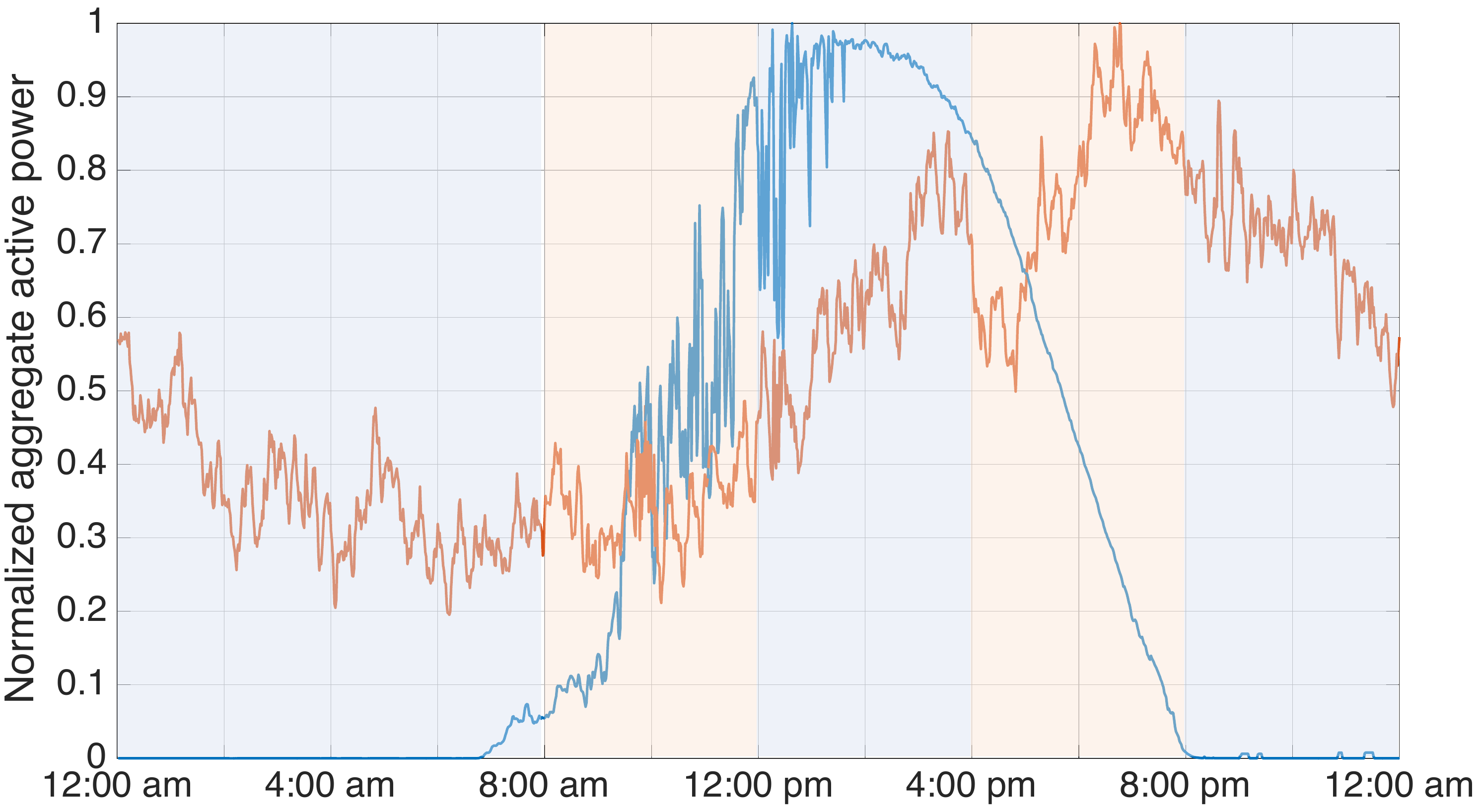}
	\caption{Normalized aggregate active load and solar generation over time. The 5 panes represent the operating periods $\mcT_1-\mcT_5$.}
	\label{fig:loadgen}
\end{figure}

The developed DNR was tested on a modified version of the IEEE 37-bus benchmark feeder converted to its single-phase equivalent~\cite{guido2018analytics}; see Fig.~\ref{fig:37bus}. Switches include three existing and two additional lines, all denoted as dashed edges. Regulator $(799,701)$ is assumed to be remotely controlled. The regulator added on line $(704,720)$ is set locally controlled with reference voltage $1$~pu and bandwidth $0.016$~pu. Five PVs of equal capacity were placed at buses $\{705, 710, 718, 730, 738\}$. Residential load and solar data were extracted from the Pecan Street dataset~\cite{pecandata}: Minute-based load and solar generation data were collected for June 1, 2018. The tested feeder has $25$ buses with non-zero load. The first $75$ non-zero load buses from the dataset were aggregated every $3$ and normalized to obtain $25$ load profiles. Similarly, $5$ solar generation profiles were obtained. The normalized minute-based feeder-aggregated load and solar profiles are shown in Fig.~\ref{fig:loadgen}. 

The normalized load profiles for the $24$-hr period were scaled so the $80$-th percentile of the total load duration curve coincided with the total nominal spot load of the feeder. This scaling results in a peak aggregate load being $1.29$ times the total nominal load. Since the Pecan Street data contained no reactive power, we synthesized reactive loads by scaling the actual demand to match the nominal power factors of the IEEE 37-bus feeder. The linearized ZIP parameters of \eqref{eq:ZIP} were found using the derived (re)active load profiles for each bus and the load type from the benchmark. The motive of the watt-var control is to alleviate overvoltages in grids with high solar integration. Thus, solar data were scaled such that $75\%$ of the overall energy consumption was met from PVs. Problem \eqref{eq:DNR} was solved using YALMIP and Gurobi~\cite{YALMIP},~\cite{gurobi}, on a $2.7$~GHz Intel Core i5 computer with $8$ GB RAM.

\begin{figure}[t]
	\centering
	\includegraphics[scale=0.38]{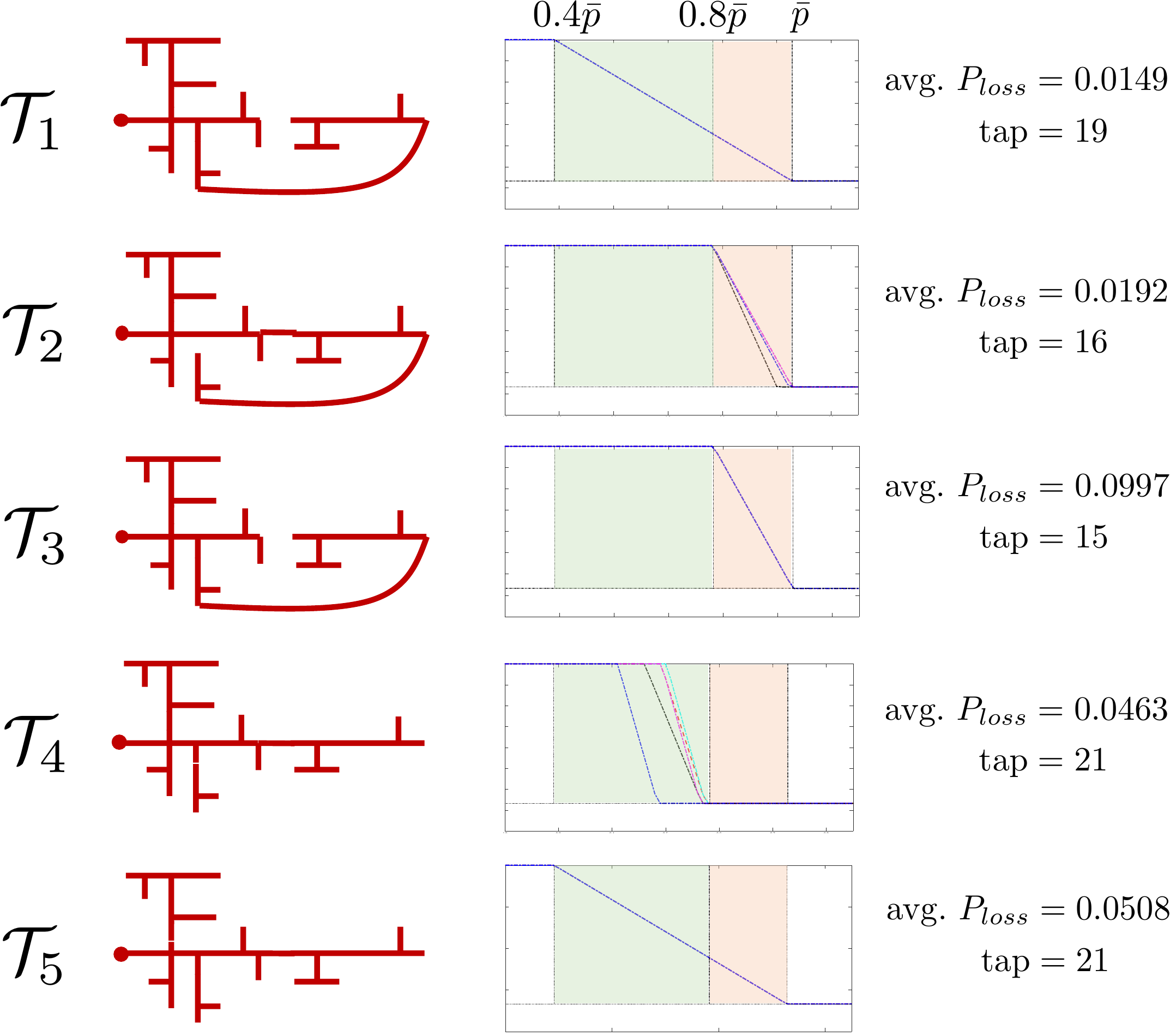}
	\caption{Results of \eqref{eq:DNR} for the feeder of Fig.~\ref{fig:37bus}. \emph{Top to bottom:} Results for the operating periods identified in Fig.~\ref{fig:loadgen}. \emph{Left to right:} optimal topology, watt-var curves for five generators, average power loss, and optimal tap position.}
	\label{fig:result}
\end{figure}

The 24-hr interval was partitioned into five periods $\mcT_1-\mcT_5$; see Fig.~\ref{fig:loadgen}. Period $\mcT_1$ extended over $8$~hr, and the rest for $4$~hr. Each period was divided into $15$-min intervals and $S=2$ samples of load and generation were randomly drawn from the minute-based data. We then solved five instances of \eqref{eq:DNR}. The results are summarized in Fig.~\ref{fig:result}. The solution times for each instance were 452,~92,~800,~268, and 50~sec. The schematics of Fig.~\ref{fig:result} depict how the optimal topology varies with changing load-generation mix throughout the day. Three distinct topologies were found to be optimal: one topology over $\mcT_1$ and $\mcT_3$; one over $\mcT_2$; and a third one over $\mcT_4-\mcT_5$. 

Period $\mcT_1$ experiences low loads and negligible solar generation. As a result, the average power loss incurred is the minimum of all periods, and hence, its watt-var curves are inconsequential. Period $\mcT_2$ features peaking generation and low load. Due to the large PV variation, a single tap setting cannot accomplish voltage regulation, and so PVs participate via reactive power absorption. Voltage regulation and loss minimization via reactive power control are known to be opposing goals~\cite{Turitsyn11}. Thus, PV generators start absorbing reactive power only when overvoltages become unavoidable. This intuition is demonstrated by the watt-var curves of Fig.~\ref{fig:result} for $\mcT_2$, where reactive absorption begins only after PVs inject more than $0.8$ of their capacity. While all PVs tend to absorb minimal reactive power and hence hit the limits of the watt-var curve in \eqref{eq:plim}, the PV at bus $738$ obtains a different curve and absorbs its maximum reactive power before reaching its $\bar{p}$. During $\mcT_3$, voltages remain within limits because both load and generation are high, and so watt-var curves coincide with minimal reactive absorption. Period $\mcT_4$ witnesses a steep decline in generation while the load is high. Therefore, a high tap setting of $21$ is needed to avoid undervoltages after the decline in generation. However, for the tap setting of $21$, reactive absorption is needed to avoid overvoltages during high PV generation. Finally, period $\mcT_5$ with no PV generation yields generic watt-var curves similar to $\mcT_1$, but higher taps and different topology from $\mcT_1$ due to high load. The average active power loss for all periods follows the loading conditions.

We also experimented with the number of operating periods and the number of samples $S$. The effects are on three fronts: \emph{i)} Frequency of changes in taps, topology, and inverter settings; \emph{ii)} Violation of voltage limits over \emph{all} minute-based data after fixing $\bomega_1$; and \emph{iii)} Total active power loss for \emph{all} minute-based data after fixing $\bomega_1$. Shorter periods inherently result in more frequent operations on taps, switches, and inverter settings, besides the communication overhead. Longer operating periods on the other hand, may render problem \eqref{eq:DNR} infeasible due to extreme changes in the load-generation mix. For instance, while an $8$-hr period for $\mcT_1$ yields an acceptable solution for \eqref{eq:DNR}, merging $\mcT_3$ and $\mcT_4$ results in infeasibility. Further, even when \eqref{eq:DNR} is feasible for a longer period, the total losses increase. Given a length, the periods should be chosen based on disparate load-generation levels. Further, for a fixed length, increasing $S$ results in lower overall losses and less voltage violations at the cost of higher computational burden, so $S$ should be determined based on the anticipated fluctuations.

\section{Conclusions}
Leveraging the automation capabilities of forthcoming active distribution grids, this work has put forth an optimal DNR approach. DERs operate under watt-var control curves to save on cyber resources. These curves are optimized jointly with the feeder topology and regulator settings. The approach uniquely integrates legacy devices and ensures radiality through intuitive and efficient optimization models. Numerical tests have corroborated: \emph{a)} The optimal topology varies with the load-generation mix; \emph{b)} Coordinating DERs and regulators is critical during periods of steep transitions; and \emph{c)} The trade-offs involved in the length of operating periods and the number of scenarios. Some open research directions are discussed next. Although this work has considered a single-phase feeder, the models should be extendable to unbalanced multiphase setups. This work has adopted a linearized grid model; the operational benefits \emph{vis-a-vis} the possible computational challenges of an AC grid model have to be explored. Although substituting watt-var with volt-var curves might seem straight-forward, the related optimization and stability issues have to be addressed. It is worth adding that the developed toolbox of radiality, DER, and regulator models is applicable when coping with other grid reconfiguration or restoration tasks. 

\balance
\bibliography{myabrv,power}
\bibliographystyle{IEEEtran}

\end{document}